\documentclass[journal]{IEEEtran}
\usepackage{graphicx}
\usepackage{amsfonts}
\usepackage[cmex10]{amsmath}
\usepackage{cite}
\usepackage{algorithmic}
\usepackage[center]{caption}
\usepackage{epsfig}
\usepackage{latexsym}
\usepackage{epstopdf}
\usepackage{verbatim}
\usepackage{amsthm}
\usepackage[linesnumbered,ruled,vlined]{algorithm2e}

\graphicspath{{./figures/}}

\begin{document}
\title{Sensing/Decision-Based Cooperative Relaying Schemes With Multi-Access Transmission: Stability Region And Average Delay Characterization}
\author{\IEEEauthorblockN{$^\dagger$Mohamed Salman, $^\dagger$Amr El-Keyi, $^\dagger$Mohammed Nafie, and  $^\star$Mazen Omar Hasna} \\
\IEEEauthorblockA{$^\dagger$Wireless Intelligent Networks Center (WINC), Nile University, Giza, Egypt \\
$^\star$College of Engineering, Qatar University Doha, Qatar\\
{\tt Mohamed.Salman@nileu.edu.eg,
\{aelkeyi,mnafie\}@nileuniversity.edu.eg, hasna@qu.edu.qa}}}

\maketitle

\begin{abstract}

We consider a cooperative relaying system which consists of a
number of source terminals, one shared relay, and a common
destination with multi-packet reception (MPR) capability. In this
paper, we study the stability and delay analysis for two
cooperative relaying schemes; the sensing-based cooperative (SBC)
scheme and the decision-based cooperative (DBC) scheme. In the SBC
scheme, the relay senses the channel at the beginning of each time
slot. In the idle time slots, the relay transmits the packet at
the head of its queue, while in the busy one, the relay decides
either to transmit simultaneously with the source terminal or to
listen to the source transmission. The SBC scheme is a novel paradigm that utilizes the spectrum more efficiently than the other cooperative schemes because the relay not only exploits the idle time slots, but also has the capability to mildly interfere with the source terminal. On the other hand, in the DBC
scheme, the relay does not sense the channel and it decides either
to transmit or to listen according to certain probabilities.
Numerical results reveal that the two proposed schemes outperform
existing cooperative schemes that restrict the relay to send only
in the idle time slots. Moreover, we show how the MPR capability
at the destination can compensate for the sensing need at the
relay, i.e., the DBC scheme achieves almost the same stability
region as that of the SBC scheme. Furthermore, we derive the
condition under which the two proposed schemes achieve the same
maximum stable throughput.

\end{abstract}

%\begin{keywords}
%Cooperative communications, multi-access relay channel, stability analysis, delay analysis, queueing %theory.
%\end{keywords}

\IEEEpeerreviewmaketitle
\section{Introduction}
%\vspace{-10pt}

\IEEEPARstart{U}{ser} demands for high data rates and services are
expected to increase exponentially in the next decade. According
to the Federal Communication Commission (FCC), about 70$\%$ of the
allocated spectrum in the US is not efficiently utilized. Hence,
effective utilization of the available spectrum is a critical
issue that has recently gained great attention. Cognitive radios
and cooperative diversity have emerged as promising techniques to
improve the wireless network performance in an attempt to exploit
the unutilized spectrum
\cite{mitola1999cognitive}--\cite{laneman2004cooperative}.
Intuitively, by relaying the messages and emptying the queues of
the primary sources, the secondary node creates more opportunities
for its own transmission.

Cooperative diversity is a new paradigm for wireless networks, and
hence, deep investigation is needed to fully understand the impact
of this new paradigm on different network layers. Most of the work
on cooperative communication has focused on the physical layer
aspects of the problem \cite{sendonaris2003user}. Other works
\cite{sadek2007cognitive}--\cite{rong2009cooperation}, however,
have implemented cooperation at the network protocol level, and
performance gains in terms of stable throughput, average delay,
and energy efficiency were illustrated. In
\cite{sadek2007cognitive}, the authors proposed a novel cognitive
multiple-access strategy, where the relay exploits the bursty
nature of the transmission of the source terminals via utilizing
their periods of silence to enable cooperation. In this strategy,
no extra channel resources are allocated for cooperation, and
hence, this improve the spectral efficiency. Although the
proposed strategy provides significant performance gain over
conventional relaying strategies, the relay is restricted to
transmit only in the idle time slots. However, allowing the relay
to send simultaneously with the source with certain probability
can improve the network performance.

Utilization of multi-packet reception (MPR) capability has
received considerable attention in the literature. An MPR model
was first introduced in \cite{ghez1988stability}. Multi-access
channel (MAC) systems with MPR capability have been addressed in
the literature in different contexts
\cite{parandehgheibi2008information}--\cite{naware2005stability}.
However, most of which do not deal with cognitive or cooperative
systems. In \cite{krikidis2010stability}, a MAC network with two
primary users, a cognitive relay, and a common destination was
considered with a symmetric configuration. The primary users,
simultaneously, access the channel to deliver their packets to a
common destination, i.e., the relay and the destination have MPR capability. The authors assume that the relay perfectly
senses the channel, i.e., it transmits during idle time slots,
where the primary nodes are not transmitting. In this scheme, the
authors assume that the primary users transmit simultaneously in
each time slot which may cause degradation in the network
performance especially in the presence of weak channels. Moreover, as the number of primary users increases, the complexity of the relay and the destination increases because these nodes have to decode the message of all the transmitting nodes in each time slots. In
\cite{fanous2013stable}, the performance of an Ad-Hoc secondary
network with $N$ secondary nodes accessing the spectrum licensed
to a primary node was demonstrated. Both cases of perfect and
imperfect sensing were considered. In the perfect sensing case,
the secondary nodes do not interfere with the primary node and
thus do not affect its stable throughput. However, with imperfect
sensing, the secondary nodes control their transmission
parameters, such as the power and the channel access
probabilities, to limit the interference on the primary node. To
compensate for the effect of interference, the authors explore the
use of the secondary nodes as relays for the primary node traffic.

The average delay encountered by the packets is one of the most
important metrics in evaluating the performance of wireless
networks. In \cite{rong2012cooperative}, the delay analysis for a
network consists of one primary user, one cognitive relay, and a
common destination was presented, where the relay is restricted to
send only in the idle time slots with full priority for the
relaying queue. Moreover, in \cite{ashour2013cooperative}, the
delay analysis for the same network with a randomized cooperative
policy was investigated, where the relay node serves either its
own data or the primary packets with certain service
probabilities. The authors, in \cite{ashour2013cooperative}, showed that the randomized policy enhances the cognitive relay delay at the expense of a slight
degradation in the primary user one.

In this paper, we consider a general number of source terminals,
one half duplex relay, and a common destination that has MPR
capability, i.e., the destination can decode the message of more
than one transmitting node in the same time slot. We consider a
slotted time division multiple access (TDMA) framework in which
each time slot is assigned to one source terminal only. We propose
two cooperative schemes. The first scheme is the sensing-based
cooperative (SBC) scheme, where the relay senses the channel at
the beginning of each time slot. If the relay detects an idle time
slot, the relay transmits the packet at the head of its queue.
Alternatively, if the relay detects a busy time slot, it decides
probabilistically to either listen to the source packet and store
it if the destination fails to decode it successfully or to
interfere with the source terminal transmission. We optimize this
probabilistic scheme to maximize the network aggregate throughput
and characterize the stability region. The main difference between
this scheme and that in \cite{sadek2007cognitive} is that, in
\cite{sadek2007cognitive}, the authors restricted the relay to
send only in the idle time slot. Moreover, unlike
\cite{krikidis2010stability}, the relay interferes with the source
terminals with certain probability to limit the adverse effects of
the interference. 

In the SBC scheme, the relay depends on the
sensing information to decide either to listen or to transmit.
Although we do not take into consideration the sensing errors and
the consumption of power, in practice these factors can negatively
affect the performance of the system. Hence, we propose the
decision-based cooperative (DBC) scheme, which is the second
cooperative scheme proposed in this work. In the DBC scheme,
unlike \cite{fanous2013stable} that assumes imperfect sensing, the
relay does not sense the channel and it decides, according to
certain probabilities, either to listen to the source or to
transmit whether the time slot is idle or busy. It is worth noting that the complexity of our proposed schemes does not increase as the number of source terminals increases, unlike \cite{krikidis2010stability}, because for any number of source terminals the destination decodes, at most, the packets of two transmitting nodes; one of the source terminals and the relay.

In this work, we focus on the medium-access layer and address the
impact of the proposed schemes on multiple-access performance
metrics such as the stable throughput region and average delay. We
show that in the two proposed schemes, the queues of the source
terminals and those of the relay are interacting. Since the
stability analysis for more than two interacting queues is
difficult, we resort to a stochastic dominance approach. The
stability analysis of interacting queues was initially addressed
in \cite{tsybakov1979ergodicity}, and later in
\cite{rao1988stability}, where the dominant system approach was
explicitly introduced. The average delay is also an important
performance measure, and its analysis illustrates the fundamental
trade-off between the rate and reliability of communication. Delay
analysis for interacting queues is a notoriously hard problem that
has been investigated in \cite{sidi1983two} and in
\cite{naware2005stability} for ALOHA with MPR channels. Hence, we
also utilize stochastic dominance to approximate the average delay
of the proposed schemes.

Our contributions in this paper can be summarized in the following
points
\begin{itemize}
\item For the two proposed schemes, we present the stability
analysis and derive the stability conditions for each queue in the
system. We formulate an optimization problem to maximize the
weighted aggregate stable throughput of the network and
characterize the stability region via optimizing the probability
of each action taken by the relay. The problem is formulated as
non-convex quadratic constrained quadratic programming (QCQP)
optimization problem \cite{huang2014randomized}. We use the
feasible point pursuit-successive convex approximation (FPP-SCA)
algorithm \cite{mehanna2014feasible} to achieve a good feasible
solution by approximating the non-convex constraints as linear
ones. \item We analyze the average delay performance for the two
proposed schemes and derive approximate delay expressions using
the dominant system approach. \item We show that the SBC scheme
provides significantly better performance over existing
cooperative schemes as in \cite{sadek2007cognitive}. Moreover, the
SBC scheme exploits the unutilized spectrum more efficiently than
other schemes, because the relay not only transmits in the idle
time slots but also has the capability to, simultaneously,
transmit its packets with the source terminals. The relay uses
this new attribute in a mild way to mitigate the negative effects
of the interference and to enhance the maximum aggregate stable
throughput of the network.  \item We demonstrate that the DBC
scheme, in certain cases, achieves the same stability region
achieved by the SBC scheme. We also illustrate how the MPR
capability at the destination can compensate for the need for the
relay to detect the idle time slots to transmit its packets.
Furthermore, we show that removing the MPR capability from the
destination, in absence of sensing at the relay, causes
catastrophic degradation in the performance of the system. \item
We derive the channel condition under which the two proposed
schemes achieve the same maximum stable throughput. Under this
condition, sensing the channel by the relay becomes useless, and
hence, the two proposed schemes provide the same performance.
\end{itemize}

The remainder of the paper is organized as follows. In Section
\ref{system_model}, we describe the system model. The SBC scheme
with its stability and delay analysis is introduced in Section
\ref{SBC}, followed by the DBC scheme in Section \ref{DBC}.
Numerical results are then presented in Section
\ref{numerical}. We demonstrate how the MPR capability
at the destination can compensate for the relay need to sense
the channel in Section \ref{discussion}.
Finally, the paper is concluded in Section \ref{conclusion}.

%and discussion are then presented in Section
%\ref{numerical} and Section \ref{discussion}, respectively.
%Finally, the paper is concluded in Section \ref{conclusion}.

\section{System Model}
\label{system_model}

We consider the uplink of a TDMA system that consists of $M $
source terminals  $ \{s_i\}_{i{=}1}^M$, one shared relay ($r$),
and one common destination ($d$), as shown in Fig. \ref{model}. We
define the set of the transmitting nodes $T{=} \{S$, $r \}$, where
$S{=} \{s_1,..., s_M\}$ is the set of source terminals, and that
of the receiving nodes $L{=}\{r,d\} $. The source terminals access
the channel by dividing the available resources among them, i.e.,
time slots in this case. Each terminal is allocated a fraction of
the time. Let $w_i$ denote the fraction of time allocated to the
source terminal $s_i$, where $ i \in \{1,2,.....,M\}$. We assume
continuous values of the resource sharing vector $\boldsymbol{w}
=[w_1,w_2,. . .,w_M]$. Hence, we can define the set of all
feasible resource sharing vector as follows
\begin{equation}
A= \Big\{  \boldsymbol{w}=[w_1,w_2,. . .,w_M] \in R_+^M, \enspace
\sum\limits_{i=1}^{M} w_i=1  \Big\}.
\end{equation}

\begin{figure}[t]
  \centering
\includegraphics[width=0.59\linewidth]{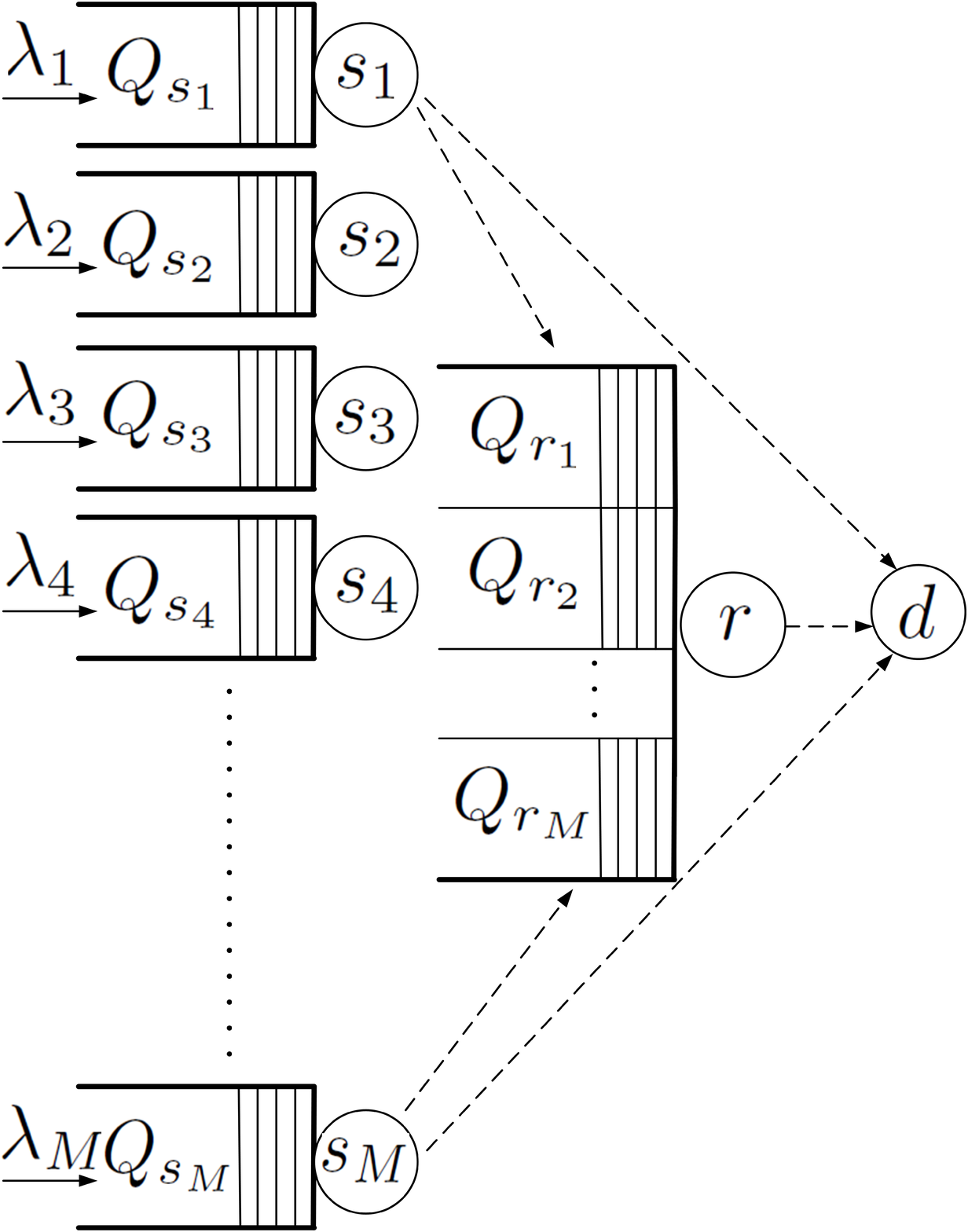}
\caption{ System Model    \label{model}}
\end{figure}

First, we describe the physical layer model. All wireless links are assumed to be stationary, frequency non-selective, and Rayleigh block fading.
The fading coefficients $h_{m,n}$, where $m  {\in} T $ and $ n {\in} L $, are assumed to be constant during each slot duration, but change independently from one time slot to another according to a circularly symmetric complex Gaussian distribution with zero mean and variance $\rho_{m,n}^2 $. All
wireless links are corrupted by additive white Gaussian noise (AWGN) with zero mean and unit variance.
All nodes transmit with fixed power $P$. An outage occurs when the instantaneous capacity of the link ($m$, $n$) is lower than the transmission rate $R$. Each link is characterized by the probability
\begin{equation}
\label{f}
f_{mn}= \mathbb{P} \{  R < \log_{2}(1+P | h_{m,n}| ^{2} ) \} = \exp  \bigg({-}\dfrac{2^{R}-1}{P \rho_{m,n}^{2}}   \bigg)
\end{equation}
which denotes the probability that the link ($m$, $n$) is not in outage. Let $g_{mn}^{I}$ denote the probability that the link
($m$, $n$) is not in outage in presence of interference from node
$I$, where $m, I  \in T $ and $n \in L$. In
\cite{krikidis2010stability}, the authors derived the same term
but for symmetric configuration. We relax this assumption and
re-derive the term to fit our model in Appendix A.

Next, we describe the medium access control layer model. Time is slotted with fixed slot duration, and the transmission of a packet takes exactly one time slot. Each source terminal has an infinite buffer (queue) to store its own incoming packets. Packet arrivals of the
source terminals are independent and stationary Bernoulli processes with
means $\lambda_i$ (packets per slot), where $i \in \{1,2, . . . ., M\}$.
The relay has $M$ relaying queues ($Q_{r_1}$, $Q_{r_2}$, . . ., $Q_{r_M}$) to store the packets of the source terminals that are not successfully decoded at the destination. Let $Q_l^t$ denote the
number of packets in the $l$-\textit{th} queue at the beginning of time slot
$t$. The instantaneous evolution of the $l$-\textit{th} queue length is given by
\begin{equation}
Q_{l}^{t+1}= (Q_{l}^t -Y_{l}^t)^{+}+X_{l}^t
\end{equation}
where $l \in \lbrace s_1,...s_M,r_1,...,r_M \rbrace$ and $(x)^{+}{=}\max
\{x,0\}$. The binary random variables $Y_{l}^t$ and $X_{l}^t$,
denote the departures and arrivals of $Q_l$ in time slot $t$,
respectively, and their values are either 0 or 1.

\section{Sensing-Based Cooperative Scheme }

\label{SBC} In this section, we introduce the proposed SBC scheme
followed by its stability and delay analysis. We assume that the
relay can sense the channel, and perfectly detect the idle time
slots. Moreover, we assume that the errors and delay in packet
acknowledgement (ACK) are negligible, which is reasonable for
short length ACK packets as low rate codes can be employed in the
feedback channel \cite{sadek2007cognitive}. In the SBC scheme, the
source terminals operate according to the following rules:
\begin{itemize}
\item Each terminal transmits the packet at the head of its queue in its assigned time slot, whenever the queue is not empty.%, if the terminal's queue is empty the slot is free.
\item If the destination decodes the packet successfully, it sends an ACK which can be heard by both the transmitting terminal and the relay. The terminal drops this packet upon hearing the ACK.
\item If the destination does not receive the packet successfully but the relay does, then the relay stores this packet at the corresponding relaying queue and sends an ACK to the source terminal. Afterwards, the relay is responsible for conveying this packet to the destination.
\item If a packet is received successfully by either the destination or the relay, the packet is removed from the terminal's queue. Otherwise, the source terminal retransmits the packet in its next assigned time slot.
\end{itemize}
It is worth noting that if the relay and the destination decode a packet successfully, the relay does not store this packet in its queue because the packet is already delivered to the destination.

Next, we illustrate the transmission and reception policy of the relay using the SBC scheme. Let us consider the time slots allocated to the source terminal $s_i$. The relay takes one of the following two actions after sensing the channel to detect the queue state of $Q_{s_i}$.
\begin{itemize}
\item If $Q_{s_i}$ is empty, i.e. $s_i$ is not transmitting, the relay transmits the packet at the
head of $Q_{r_i}$. \item If $Q_{s_i}$ is not empty, the relay
transmits a packet from $Q_{r_j}$, where $j \in \{1,2,...,M\}$,
with probability $\beta_{ij}$ or listens to source terminal
transmission with probability
$1{-}\sum\limits_{j=1}^{M}\beta_{ij}$.
\end{itemize}
Hereafter are some important remarks. %First, if the relay detects a busy time slot, a random experiment is performed to decide the operation policy according to the values of $\beta_{ij}$. Second, the relay stores the packet of the source terminal if it decides to listen to the terminal's transmission and the destination fails to decode the packet successfully. Third, we assume
First, the relay stores the packet of the source terminal if it
decides to listen to the terminal's transmission and successfully
decodes the packet while the destination fails to decode it.
Second, we assume that the relay is half duplex, and hence, it can
not transmit and receive a packet in the same time slot. Thus, the
relay can not receive packets when it decides to interfere with
the source terminal.

From the given description of the proposed cooperative scheme, it is clear that the decision taken by the relay depends on the queue state of the source terminals, and this causes an interaction between the relay's queues and those of the source terminals. Stability of interacting queues is complex \cite{tsybakov1979ergodicity}, thus, we resort to using the  dominant system approach. To perform the stability analysis of $s_i$, we assume that all other source terminals $s_j$, where $j \neq i $, are saturated, i.e, their queues always have packets. Moreover, we assume that $Q_{r_i}$ transmits a dummy packet whenever the relay decides to interfere with the source transmission and $Q_{r_i}$ is empty, where $i\in \{1,2,...,M\}$. This dominant system simplifies the stability analysis and provides an inner bound on the stability region.
%Moreover, $Q_{s_j}$ transmits a
%dummy packet whenever it is empty and the relay decides to transmit, simultaneously with $s_j$, a packet
%from $Q_{r_i}$, where $i \neq j$.
%In the second dominant system, $Q_{r_i} $ transmits a dummy packet in the same way as in the first one, while, $Q_{s_j}$ transmits a dummy packet whenever the relay interferes with $s_j$ for all $j = \{1,2,...,M\}$. This dominant system provides a upper bound on the delay performance.

\subsection{Stability Region Analysis}
\label{throughput}

In this part, we characterize the stability region for the
SBC scheme taking in consideration the dominant system described above. Then, we formulate
an optimization problem to maximize the aggregate throughput and characterize the stability region via
optimizing the values of $ \{ \beta_{ij} \}_{i,j{=}1}^{M} $ subject to constraints that ensure the stability of the system.
The stability of the overall system requires the stability of each
individual queue. From the definition in \cite{szpankowski1994stability}, the queue is stable if
\begin{equation}
  \lim_{t\to \infty} \mathbb{P} \{Q_l^t < x \} = F(x)\enspace   \text{and} \enspace  \lim_{x\to \infty} F(x) =1
\end{equation}
We can apply Loynes' theorem to check the stability of a queue
\cite{loynes1962stability}. Loynes' theorem states that if the
arrival process and the service process of a queue are strictly
stationary, then the queue is stable if and only if the average
service rate is greater than the average arrival rate of the
queue.

In the dominant system, a packet departs $Q_{s_i}$, where $i \in
\{1,2,...,M \}$ in two cases. First, if it is successfully decoded by the destination when the relay decides to interfere with $s_i$ from one of its relaying queues. Second, if it is successfully decoded by at least one node, i.e., the destination or the relay when the relay decides to listen to $s_i$. Thus, the average service rate of $Q_{s_i}$ for certain $w_i$ is given by
\begin{equation}
\label{muei_enhance}
\mu_{i}= w_i \Big(\sum\limits_{j=1}^{M} \beta_{ij} g_{s_id}^r+(1-\sum\limits_{j=1}^{M} \beta_{ij} )(f_{s_id}+(1-f_{s_id}) f_{s_ir})\Big).
\end{equation}
For stability of $Q_{s_i}$, the following condition must be
satisfied
\begin{equation}
\label{lambdai_enhance}
\lambda_{i} <  w_i \Big(\sum\limits_{j=1}^{M} \beta_{ij} g_{s_id}^r+(1-\sum\limits_{j=1}^{M} \beta_{ij} )(f_{s_id}+(1-f_{s_id}) f_{s_ir})\Big).
\end{equation}

A packet arrives at $Q_{r_i}$ if the following two conditions are met. First, an outage occurs in
the link between $s_i $ and the destination node while no outage
occurs in the link between $s_i$ and the relay and this happens
with probability $(1{-}f_{s_id})f_{s_ir}$. Second, $Q_{s_i}$ is
not empty and the relay decides to receive from $s_i$ and this happens with probability $\frac{\lambda_i}{\mu_i} (1-\sum\limits_{j=1}^{M} \beta_{ij})$. Thus, the average arrival rate of $Q_{r_i}$ is given by
\begin{equation}
\label{lambdari_enhance}
\lambda_{r_i}= w_i (1{-}f_{s_id})f_{s_ir} \frac{\lambda_i}{\mu_i} (1-\sum\limits_{j=1}^{M} \beta_{ij})
\end{equation}

In the dominant system, the packet departs $Q_{r_i}$ in three cases. First, when $Q_{s_i}$ is empty a packet departs $Q_{r_i}$ if no outage occurs in the link between the relay and the destination. Second, when $Q_{s_i}$ is not empty a packet departs $Q_{r_i}$ if the relay decides to interfere with $s_i$ with a packet from $Q_{r_i}$. Third, if the relay decides to interfere with saturated source terminal $s_j$, where $j \neq i$, with a packet from $Q_{r_i}$. Thus, the average service rate of $Q_{r_i}$ can be expressed by
\begin{equation}
\label{muri_enhance}
\mu_{r_i}{=} w_i  \bigg( \Big( 1{-}\frac{\lambda_i}{\mu_i} \Big) f_{rd} {+} \frac{\lambda_i}{\mu_i}\beta_{ii} g_{rd}^{s_i} \bigg){+}  \sum\limits_{j \neq i,  j=1}^{M}w_j  \beta_{ji} g_{rd}^{s_j}.
\end{equation}
For the stability of $Q_{r_i}$, the service rate must be higher than
the arrival rate, i.e., $\lambda_{r_i}< \mu_{r_i}$ , and hence, we have
\begin{equation}
\label{lambdaii_enhance}
\lambda_i< \frac{w_i f_{rd}+ \sum\limits_{j=1, j \neq i}^{M} w_j \beta_{ji}g_{rd}^{s_j}}{w_i \Big( (1-\sum\limits_{j=1}^{M} \beta_{ij})(1-f_{s_id})f_{s_ir}+f_{rd}-\beta_{ii}g_{rd}^{s_i} \Big)} \mu_i
\end{equation}

Let $\lambda_s^i$ denote the maximum stable throughput for the \textit{i}-th source terminal at a certain value of $w_i$. To guarantee the stability of the entire network the conditions in (\ref{lambdai_enhance}) and (\ref{lambdaii_enhance}) must be satisfied. Hence, the maximum stable throughput is given by
\begin{equation}
\label{max_throughput_SBC}
    \lambda_s^i < \min\{\mu_i,\mu_{u_i}\}
\end{equation}
where $\mu_{u_i}$ can be obtained easily from (\ref{lambdaii_enhance}) as follows
\begin{equation}
\label{muui_enhance}
\mu_{u_i}= \frac{w_i f_{rd}+ \sum\limits_{j=1, j \neq i}^{M} w_j \beta_{ji}g_{rd}^{s_j}}{w_i \Big( (1-\sum\limits_{j=1}^{M} \beta_{ij})(1-f_{s_id})f_{s_ir}+f_{rd}-\beta_{ii}g_{rd}^{s_i} \Big)} \mu_i
\end{equation}

%Next, we formulate an optimization problem to calculate the maximum
%overall stable throughput of the network subject to the stability
%of all queues. Hence, the optimization problem is given by
Next, we optimize the values of $ \{\beta_{ij}\}_{i,j{=}1}^M$ to achieve the maximum weighted aggregate throughput of the network subject to constraints that ensure the stability of all queues. Thus, the optimization problem can written\footnote{We do not restrict the second and the third constraints to be satisfied with strict inequality as in (\ref{lambdai_enhance}) and (\ref{lambdaii_enhance}). However, the solver uses the interior point method which provide a strictly feasible point.}, for certain $w_i$, as follows
\begin{equation}
\label{opt1}
\begin{aligned}
& \underset{ \enspace \enspace  \{\beta\}_{i,j{=}1}^M }{\text{max}}
&&   \sum\limits_{i=1}^{M} x_i \lambda_s^i \\
& \text{subject to}
&&  0 \leq  \sum\limits_{j=1}^{M}\beta_{ij} \leq 1,  & \enspace i \in \{1,2,...,M\} \\
&&& \lambda_s^i  \leq \mu_i , & \enspace i \in \{1,2,...,M\}\\
&&& \lambda_s^i  \leq \mu_{u_i} , & \enspace i \in \{1,2,...,M\} \\
&&& 0 \leq  \beta_{ij}  \leq 1, &  \enspace i,j \in \{1,2,...,M\}\\
\end{aligned}
\end{equation}
where $x_i$ is the weight assigned to $\lambda_s^i$. To obtain the
stability region, we have to vary the value of each $x_i$, where
$i \in \{1,2,...,M\}$, from zero to one for each $w_i$. Then vary
$w_i$ also from zero to one to scan the whole stability region,
i.e., the convex hull of all the obtained values of
$\{\lambda_s^i\}_{i=1}^M$ provides the stability region. On the
other hand, we set all values of $x_i$ by ones to calculate the
maximum aggregate stable throughput of the network.

In the above problem, the summation in the first constraint is the probability that the relay transmits a packet from one of its relaying queue while $Q_{s_i}$ is not empty. The second and the third constraints guarantee the stability of the queues in the network. In order to solve this optimization problem, we define a new
$M(M+1)$-dimensional vector $ \boldsymbol{\beta} {=}[\beta_{11},..., \beta_{1M}, ...,
\beta_{M1},...,\beta_{MM}, \lambda_s^1,..., \lambda_s^M]^T$ and
rewrite the optimization problem, in the standard form
\cite{boyd2004convex}, as follows
\begin{equation}
\label{opt1_modified_SBC}
\begin{aligned}
& \underset{\boldsymbol{\beta}}{\text{min}}
&&   \boldsymbol{x}^T \boldsymbol{\beta} \\
& \text{s.t.}
&&  0 \leq  \boldsymbol{b}_i^T \boldsymbol{\beta}  \leq 1 ,  & \enspace i \in \{1,2,...,M\} \\
&&& \boldsymbol{v}_i^T \boldsymbol{\beta}+u_i \leq 0 ,  & \enspace i \in \{1,2,...,M\}  \\
&&&  \boldsymbol{\beta}^T \boldsymbol{A}_i \boldsymbol{\beta} + \boldsymbol{c}_i^T \boldsymbol{\beta}+d_i \leq 0 ,  & \enspace i \in \{1,2,...,M\} \\
&&&  \boldsymbol{0} \preceq \boldsymbol{\beta} \preceq
\boldsymbol{1}
\end{aligned}
\end{equation}
where $x{=} [0,0,...,0,-x_1,...,-x_M]$, and the terms
$\boldsymbol{v}_i$, $u_i$, $\boldsymbol{A}_i$, $\boldsymbol{c}_i,$
and $d_i$ can be easily obtained from (\ref{muei_enhance}) and
(\ref{muui_enhance}). From (\ref{muei_enhance}), the stability of
$Q_{s_i}$ is represented by linear constraints in
$\boldsymbol{\beta}$ while, from (\ref{muui_enhance}), the
stability of $Q_{r_i}$ is represented by quadratic constraints in
$\boldsymbol{\beta}$. The objective and the linear constraints are
convex. However, the quadratic constraints are not convex because
$\boldsymbol{A}_i$ is an indefinite matrix. In general, non-convex
QCQP problems are NP hard \cite{vandenberghe1996semidefinite},
except for special cases such as those in \cite{huang2010rank}.
Several methods have been proposed to approximate non-convex QCQP
problems, including semi-definite relaxation (SDR)
\cite{vandenberghe1996semidefinite}, the reformulation
linearization technique (RLT) \cite{sherali1998reformulation}, and
successive convex approximation (SCA) \cite{beck2010sequential}.
In our case, we use an iterative algorithm to obtain a good
feasible solution as in \cite{mehanna2014feasible}. We approximate
the feasible region through a linear restriction of the non-convex
parts of the constraints. The solution of the resulting
optimization problem is then used to compute a new linearization
and the procedure is repeated until convergence. Using the
eigenvalue decomposition, the matrix $\boldsymbol{A}_i$ can be
expressed as
$\boldsymbol{A}_i=\boldsymbol{A}_i^++\boldsymbol{A}_i^-$, where
$\boldsymbol{A}_i^+ \succeq 0$ and $\boldsymbol{A}_i^-\preceq 0$.
For any $\boldsymbol{y} \in R^{M(M+1)\text{x}1}$, we can replace
the non-convex constraint in (\ref{opt1_modified_SBC}) by the
following convex one
\begin{equation}
\boldsymbol{\beta}^T \boldsymbol{A}_i^+ \boldsymbol{\beta} +2
\boldsymbol{y}^T \boldsymbol{A}_i^- \boldsymbol{\beta} +
\boldsymbol{c}_i^T \boldsymbol{\beta}+d_i \leq \boldsymbol{y}^T
\boldsymbol{A}_i^-\boldsymbol{y}
\end{equation}
See \cite{mehanna2014feasible} for more details. Thus, the non-convex problem is converted to a convex one, and we use Algorithm 1 to solve the optimization problem in (\ref{opt1}).
\begin{algorithm}[h]
\label{lago}
\caption{We use the FPP-SCA algorithm to achieve good feasible point}
\textbf{Initialization}: set $k=0$ and $ \boldsymbol{y}_0= \boldsymbol{0}$. \\
\textbf{Repeat} \\
 \begin{enumerate}
 \item solve
 \begin{equation*}
 \label{opt1_final}
 \begin{aligned}
 & \underset{\boldsymbol{\beta}}{\text{min}}
 &&   \boldsymbol{x}^T \boldsymbol{\beta} \\
 & \text{s. t.}
&&  0 \leq  \boldsymbol{b}_i^T \boldsymbol{\beta}  \leq 1 \\
&&& \boldsymbol{v}_i^T \boldsymbol{\beta}+u_i \leq 0  \\
&&&  \boldsymbol{\beta}^T \boldsymbol{A}_i^+ \boldsymbol{\beta} +2 \boldsymbol{y}_k^T \boldsymbol{A}_i^-\boldsymbol{\beta} + \boldsymbol{c}_i^T \boldsymbol{\beta}+d_i \leq \boldsymbol{y}_k^T \boldsymbol{A}_i^-\boldsymbol{y}_k \\
&&&  0 \leq \boldsymbol{\beta} \leq 1,  \enspace i \in \{1,2,...,M\}
 \end{aligned}
 \end{equation*}
 \item Let $\boldsymbol{\beta}_k^*$ denote the optimal $\boldsymbol{\beta}$ obtained at the $k$-\textit{th} iteration, and set  $\boldsymbol{y}_{k+1}=\boldsymbol{\beta}_k^*$
 \item Set $k=k+1$. \\
 \end{enumerate}
 \textbf{until convergence}.
\end{algorithm}

 It is worth noting that we can use other techniques to solve this problem, but the FFP-SCA obtains good feasible point even for very large $M$. Moreover, a few iterations are required for convergence \cite{mehanna2014feasible}.

\subsection{Average Delay Characterization}

In this subsection, we characterize the average delay for the
dominant system of the SBC scheme. We derive an approximate
expression of the delay as in \cite{sadek2007cognitive}, where we
assume that the relay queues are discrete-time M/M/1 queues. If the packet is directly transmitted
from the source terminal to the destination, it experiences a
queueing delay only in the terminal's queue. Alternatively, if the
packet is delivered to the destination through the relay, this
packet experiences two queuing delay; one in the terminal's queue
and the other in the relay's queue. The packet experiences
\textit{only} the queueing delay at the source terminal with the
following probability
\begin{equation}
\epsilon_i = w_i\frac{\Big(1{-}\sum\limits_{j=1}^{M}\beta_{ij} \Big) f_{s_id} + \sum\limits_{j=1}^{M}\beta_{ij} g_{s_id}^{r}}{\mu_i}
\end{equation}
which is the probability that the packet is successfully decoded by the destination given
that it is dropped from the source terminal. Thus, the average delay encountered by the packets of the $i$-\textit{th} source terminal is given by
\begin{align}
\label{delayi_enhance}
D_i &= \epsilon_i T_{s_i}+ (1-\epsilon_i) (T_{s_i}+T_{r_i})  \nonumber \\
&= T_{s_i}+(1-\epsilon_i)T_{r_i}  \enspace \enspace
\end{align}
where $T_{s_i}$ and $T_{r_i}$ denote the average queueing delays at $s_i$
and $r_i$, respectively. Since the arrival rates at $Q_{s_i}$ and $Q_{r_i}$ are
given by $\lambda_i$ and $\lambda_{r_i}$, respectively, then applying Little's law
yields
\begin{equation}
\label{ti_enhance}
        T_{s_i}= N_i / \lambda_i       ,   \enspace        T_{r_i} = N_{r_i} / \lambda_{r_i}
\end{equation}
where $N_i$ and $N_{r_i}$ denote the average queue size of $Q_{s_i}$ and $Q_{r_i}$, respectively. The queues of the source terminals are discrete-time M/M/1 queues with
Bernoulli arrivals and Geometrically distributed service rates, and we assume that $Q_{r_i}$ is a discrete-time M/M/1 queue. Thus, we can easily calculate $N_{s_i}$ and $N_{r_i}$, by applying the Pollaczek-Khinchine formula \cite{kleinrock1975queueing}, as follows

%We can easily calculate $N_{s_i}$ and $N_{r_i}$ by observing
%that $Q_{s_i}$ and $Q_{r_i}$ are discrete-time M/M/1 queues with
%Bernoulli arrivals and Geometrically distributed service rates.
%Then, by applying the Pollaczek-Khinchine formula for discrete time \cite{kleinrock1975queueing}, we obtain $N_i$ and $N_{r_i}$ as follows
\begin{equation}
        N_i= \dfrac{-\lambda_i^2 + \lambda_i}{\mu_i-\lambda_i} , \enspace \enspace N_{r_i}= \dfrac{-\lambda_{r_i}^2 + \lambda_{r_i}}{\mu_{r_i}-\lambda_{r_i}} \enspace
        \label{ni_enhance}
\end{equation}

Substituting (\ref{ti_enhance}) and (\ref{ni_enhance}) in
(\ref{delayi_enhance}), we can write the average queueing delay
for the $i$-th source terminal as
\begin{equation}
\label{average_delay_enhance}
D_i = \frac{1-\lambda_i}{\mu_i-\lambda_i}+ (1-\epsilon_i) \frac{1-\lambda_{r_i}}{\mu_{r_i}-\lambda_{r_i}}
\end{equation}
where the expressions of $\mu_i$ and $\lambda_{r_i}$, and $\mu_{r_i}$ are obtained in (\ref{muei_enhance}), (\ref{lambdari_enhance}), and (\ref{muri_enhance}), respectively.

\section{Decision-Based Cooperative Scheme}
\label{DBC} In this section, we present the proposed DBC scheme
together with its stability and delay analysis. In this scheme,
the behavior of the source terminals is exactly the same as in the
SBC scheme, where each source terminal transmits only in its
assigned time slots. The source drops the packet if it hears an
ACK from the relay or the destination, otherwise, the source
retransmits the packet in the next assigned time slot. The main
difference between the DBC and the SBC schemes is in the way the
relay operates. In the SBC scheme, the relay decides its operation
policy depending on the sensing information. However, in the DBC
scheme, the relay does not sense the channel and it decides
\textit{randomly} either to transmit or to listen regardless the
queue state of $s_i$. To illustrate the operation policy of the
relay in the DBC scheme, let us consider the time slots allocated
to $s_i$ where the relay takes one of the following actions:
\begin{itemize}
\item The relay transmits a packet from $Q_{r_j}$ with probability
$\alpha_{ij}$, where $j \in \{1,2,..,M\}$ \item The relay listens
to the source transmission with probability
$1-\sum\limits_{j=1}^{M} \alpha_{ij} $. The relay stores the
source's packet if the relay successfully decodes it while the
destination fails.
\end{itemize}
%From the given description of the proposed policy, the relay decides an action
%according to the values of $\alpha_{ij}$ , i.e., a random
%experiment is performed at the beginning of each time slot
%by the relay to determine the action that it will take from the
%above list.

%Note that from the given description of the DBC policy that the decision taken by the relay does not depend on the queue state of the source terminals.
When the relay decides to transmit without listening to the
channel, the probability of successful transmission depends on the
queue state of the source terminal because if the terminal's queue
is empty, there is a higher probability that the destination
decodes the relay's packet. Hence, it is clear that there is an
interaction between the relay queues and those of the source
terminals. To simplify the analysis, we assume the same dominant
system used in the SBC scheme, where $Q_{r_i}$ transmits a dummy
packet whenever the relay, randomly, decides to transmit a packet
from $Q_{r_i}$ while this queue is empty. Furthermore, to perform
the stability analysis of a certain source terminal, we assume
that all other source terminals are saturated, i.e., their queues
are not empty at any time slot. This dominant system simplifies
the stability analysis and provides an inner bound on the
stability region.

\subsection{Stability Region Analysis}

In this subsection, we follow the same steps as in the SBC scheme.
First, we characterize the stability conditions for all the queues
in the DBC scheme. Then, we formulate an optimization problem to
maximize the weighted aggregate throughput and characterize the
stability region by optimizing the values of $
\{\alpha_{ij}\}_{i,j{=}1}^M $ under the stability constraints.

In the dominant system, a packet departs $Q_{s_i}$, where $i \in \{1,2,.....,M \}$, in two cases. First, if it is successfully decoded by at least one node, i.e., the destination or the relay when the relay decides to listen to $s_i$. Second, if it is successfully decoded by the destination when the relay decides to interfere with $s_i$. Thus, the average service rate of $Q_{s_i}$ for certain $w_i$ is given by
\begin{equation}
\label{muei_DBC}
\mu_{i}= w_i \Big( (1-\sum\limits_{j=1}^{M} \alpha_{ij})(f_{s_id}+(1-f_{s_id}) f_{s_ir} ) + \sum\limits_{j=1}^{M} \alpha_{ij} g_{s_id}^r \Big)
\end{equation}
For stability of $Q_{s_i}$, the following condition must be satisfied
\begin{equation}
\label{lambdai_DBC}
\lambda_{i} <  w_i \Big( (1-\sum\limits_{j=1}^{M} \alpha_{ij})(f_{s_id}+(1-f_{s_id}) f_{s_ir} ) + \sum\limits_{j=1}^{M} \alpha_{ij} g_{s_id}^r \Big)
\end{equation}
%where $\lambda_i$ is the average arrival rate of $Q_{s_i}$.

A packet arrives at $Q_{r_i}$ if the following conditions are met. First, an outage occurs in the link
between $s_i $ and the destination node while no outage occurs in the link
between  $s_i$ and the relay. Second, $Q_{s_i}$ is not empty which has
a probability of $\lambda_i/\mu_i$, and the relay decides to listen to $s_i$ which happens with probability $1-\sum\limits_{j=1}^{M} \alpha_{ij}$. Thus, the average arrival rate
of $Q_{r_i}$ is given by
\begin{equation}
\label{lambdari_DBC}
\lambda_{r_i}= w_i (1-f_{s_id}) f_{s_ir}  \frac{\lambda_{i}}{\mu_{i}} (1-\sum\limits_{j=1}^{M} \alpha_{ij})
\end{equation}

In the dominant system, a packet departs from $Q_{r_i}$ in two cases. First, if the relay randomly decides to transmits a packet from $Q_{r_i}$ on the time slot allocated to $s_i$. In this case the packet departs $Q_{r_i}$ with probability $ w_i \big(\alpha_{ii} f_{rd} (1- \frac{\lambda_i}{\mu_i}) {+} \alpha_{ii} g_{rd}^{s_i} \frac{\lambda_i}{\mu_i} \big)$. Second, if the relay decides to transmits a packet from $Q_{r_i}$ in the time slot allocated to the saturated source terminal $s_j$, where $j \neq i$. In this case the packet departs $Q_{r_i}$ with probability $ \sum\limits_{j{=}1,j\neq i}^{M} w_j \alpha_{ji} g_{rd}^{s_j}$. Thus, the average service rate of $Q_{r_i}$ can expressed by
\begin{equation}
\label{muri_DBC}
\mu_{r_i}{=}  w_i \Big(\alpha_{ii} f_{rd} (1- \frac{\lambda_i}{\mu_i}) {+} \alpha_{ii} g_{rd}^{s_i} \frac{\lambda_i}{\mu_i} \Big){+}\sum\limits_{j{=}1,j\neq i}^{M} w_j \alpha_{ji} g_{rd}^{s_j}
\end{equation}
For the stability of $Q_{r_i}$, the service rate must be higher than the arrival rate, i.e., $\lambda_{r_i}< \mu_{r_i}$ , and hence, we have
\begin{equation}
\label{lambdaii_DBC}
\lambda_i< \frac{ w_i \alpha_{ii} f_{rd} {+} \sum\limits_{j{=}1,j\neq i}^{M} w_j \alpha_{ji} g_{rd}^{s_j}}{w_i \Big( (1{-}\sum\limits_{j=1}^{M}\alpha_{ij})f_{s_ir} (1-f_{s_id})+\alpha_{ii} f_{rd} -\alpha_{ii} g_{rd}^{s_i}         \Big)} \mu_i
\end{equation}

To guarantee the stability of the network the following condition must be satisfied
\begin{equation}
\label{max_stable_DBC}
    \lambda_s^i < \min\{\mu_i,\mu_{u_i}\}
\end{equation}
where $\mu_{u_i}$ obtained directly from (\ref{lambdaii_DBC}) as follows
\begin{equation}
\label{mueui_DBC}
\mu_{u_i}= \frac{ w_i \alpha_{ii} f_{rd} {+} \sum\limits_{j{=}1,j\neq i}^{M} w_j \alpha_{ji} g_{rd}^{s_j}}{w_i \Big( (1{-}\sum\limits_{j=1}^{M}\alpha_{ij})f_{s_ir} (1-f_{s_id})+\alpha_{ii} f_{rd} -\alpha_{ii} g_{rd}^{s_i}         \Big)} \mu_i
\end{equation}

We can write the optimization problem to calculate the maximum
weighted stable throughput of the network as
\begin{equation}
\label{opt2_DBC}
\begin{aligned}
& \underset{ \enspace \enspace \{\alpha_{ij}\}_{i,j{=}1}^M  }{\text{max}}
&&  \sum\limits_{i=1}^{M} x_i \lambda_s^i \\
& \text{subject to}
&&  0 \leq  \sum\limits_{j=1}^{M} \alpha_{ij} \leq 1,  & \enspace i \in \{1,2,......,M\} \\
&&& \lambda_s^i  \leq \mu_i  , & \enspace i \in \{1,2,......,M\}  \\
&&& \lambda_s^i  \leq \mu_{u_i} , & \enspace i \in \{1,2,......,M\}   \\
&&& 0 \leq  \alpha_{ij}  \leq 1, &  \enspace i,j \in \{1,2,......,M\} \\
\end{aligned}
\end{equation}
This optimization problem is exactly the same problem as that in the SBC scheme. Thus, we follow the same steps to solve this problem, and use the FPP-SCA algorithm to obtain good feasible solution. %to solve it as shown in Algorithm 2

%\begin{algorithm}[h]
%\label{algo2}
%\caption{We use the FPP-SCA algorithm to obtain feasible solution. }
%\textbf{Initialization}: set $k=0$ and $y_0= 0$ vector. \\
%\textbf{Repeat} \\
% \begin{enumerate}
% \item solve
% \begin{equation*}
% \label{opt5}
% \begin{aligned}
% & \underset{\beta}{\text{min}}
% &&   e^T \alpha \\
% & \text{s. t.}
%&&  0 \leq  h_i^T \alpha  \leq 1 \\
%&&& p_i^T \alpha+q_i \leq 0  \\
%&&&  \alpha^T B_i^+ \alpha +2 y_k^T B_i^-\alpha + v_i^T \alpha+z_i \leq y_k^T B_i^-y_k \\
%&&&  0 \leq \alpha \leq 1,  \enspace i \in \{1,2,...,M\}
% \end{aligned}
% \end{equation*}
% \item Let $\alpha_k^*$ denote the optimal $\alpha$ obtained at the k-\textit{th} iteration, and set  $y_{k+1}=\alpha_k^*$
% \item Set $k=k+1$. \\
% \end{enumerate}
% \textbf{until converge}.
%\end{algorithm}

%where the terms $q_i, p_i, B_i, v_i,$ and $z_i$ can be obtained from (\ref{muei_DBC}) and (\ref{mueui_DBC}). It worth recalling that to calculate the maximum aggregate stable throughput, we set all $e_i$ by one. However, to achieve the broader stability region, we vary $e_i$ and $w_i$ from zero to one to sweep the whole stability region, then take the convex hull of all the obtained values.

\subsection{Average Delay Characterization}
In the DBC scheme, the source packets are either transmitted
directly to the destination or through the relay. We define $\tau_i$, which is the probability that the packet is
successfully decoded by the destination given that it is dropped
from the $s_i$, as
\begin{equation}
\tau_i= w_i\frac{\Big(1{-}\sum\limits_{j=1}^{M}\alpha_{ij} \Big) f_{s_id} + \sum\limits_{j=1}^{M}\alpha_{ij} g_{s_id}^{r}}{\mu_i}
\end{equation}
Thus, the packets of the $i$-\textit{th} source terminal
experience the following average delay
\begin{equation}
D_i =  T_{s_i}+(1-\tau_i)T_{r_i}
\end{equation}
As in the SBC scheme, we can calculate $T_{s_i}$ and $T_{r_i}$ by
assuming that $Q_{r_i}$ is a discrete-time M/M/1 queue and
applying Pollaczek-Khinchine formula. Hence, we can write the
average queueing delay for the $i$-th source terminal's packets as
\begin{equation}
\label{average_delay_DBC}
D_i = \frac{1-\lambda_i}{\mu_i-\lambda_i}+ (1-\tau_i) \frac{1-\lambda_{r_i}}{\mu_{r_i}-\lambda_{r_i}}
\end{equation}
where the expressions of $\mu_i$, $\lambda_{r_i}$, and $\mu_{r_i}$ are obtained in (\ref{muei_DBC}), (\ref{lambdari_DBC}), and (\ref{muri_DBC}), respectively.

\section{Numerical Results}
\label{numerical}
In this section, we evaluate the performance of the two proposed schemes for the case of two source terminals, i.e., $M{=}2$. We compare the performance of the SBC and DBC schemes with other cooperative schemes. First, we compare with the cooperative cognitive multiple-access (CCMA) scheme, proposed in \cite{sadek2007cognitive}, where the authors assume that the relay only transmits in the idle time slots. %In other words, the relay transmits a packet from the head of $Q_{r_i}$ \textit{only} if $Q_{s_i}$ is empty.
Moreover, the relay only helps the terminals which, on average,
have worse channel condition than the relay itself. In other
words, the relay assists the terminals whose outage probability to
the destination satisfy $f_{rd} > f_{s_id}$, where $i \in
\{1,2,...,M\}$. Second, we compare with the DBC scheme while
setting $g_{ij}^I{=}0$, i.e., the destination can decode only if
there is one node is transmitting per time slot. We refer to this
scheme by the collision model of the DBC scheme (CM-DBC). This
scheme is important to illustrate how the MPR capability at the
destination is essential in the absence of sensing at the relay.
Third, we compare with the SBC and DBC schemes but with
$\beta_{12}{=}\beta_{21}{=}0$ and $\alpha_{12}{=}\alpha_{21}{=}0$,
respectively, to illustrate the effect of these probabilities and
demonstrate the cases where their effect on the performance is essential.
By setting $\beta_{12}{=}\beta_{21}{=}0$ and $\alpha_{12}{=}\alpha_{21}{=}0$,
the relay can not transmit a relayed packet from $Q_{r_i}$ by interfering 
on $s_j$, where $i\neq j$.

We consider three cases for the channel conditions. In the first
case, the system parameters are chosen as follows: $ P{=}10,
R{=}1$, $\rho_{s_1,d}^{2}{=}0.02$, $\rho_{s_2,d}^{2}{=}0.84$,
$\rho_{s_1,r}^{2}{=}0.97$, $\rho_{s_2,r}^{2}{=}0.93$, and
$\rho_{r,d}^{2}{=}0.03$. This case corresponds to an asymmetric
channel situation, where the $r$-$d$ channel and the $s_1$-$d$
channel are weak while $s_2$-$d$ channel is strong. In Fig.
\ref{Unsymmetric_Bad_relay}, we plot the stability region for the
cooperative schemes via varying the value of $w_1$ from zero to
one by step $0.1$. It is obvious from the figure that the SBC
scheme significantly outperforms the CCMA scheme. The rationale
behind this enhancement is that in the CCMA scheme the relay is
restricted to transmit its packets in the idle time slots only. On
the other hand, the SBC scheme adds to the relay the capability
to, simultaneously, transmit its packets with the source terminal
while controlling the interference probabilistically. This
capability expands the stability region.

It is worth noting that the two proposed schemes, SBC and DBC, can
achieve exactly the same maximum stable throughput for $s_2$,
while this is not the case for $s_1$. In the DBC scheme, the relay
does not sense the channel, and hence, the relay may interfere with
the source terminals. If the relay transmits a packet without
sensing and interferes with $s_2$, the destination with high probability
decodes the packet of $s_2$ first by treating the relay's signal
as noise, since the $s_2$-$d$ channel is strong, then decodes the
relayed signal. This is not the case for $s_1$ because if the
relay interferes with $s_1$ in the presence of these weak channels,
$s_1$-$d$ and $r$-$d$, the destination most probably fails to
decode both signals. The huge performance gap between the DBC
scheme and the CM-DBC scheme demonstrates the importance of MPR
capability at the destination in the absence of sensing capability
at the relay. Removing the MPR capability causes a catastrophic
reduction in the performance.

\begin{figure}[t]
  \centering
\includegraphics[width=.95\linewidth,height=.25\textheight]{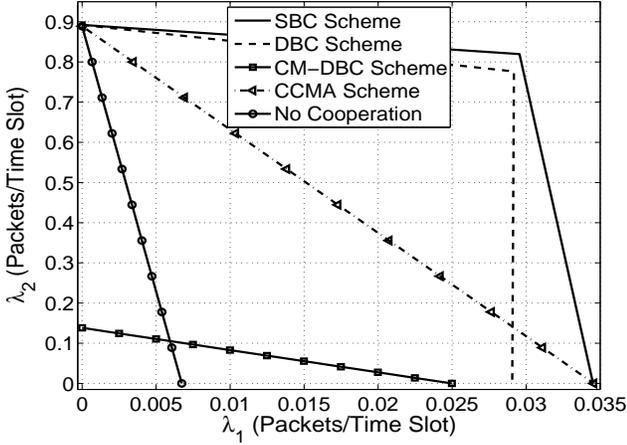}
\caption{ Stable throughput region for asymmetric channels configuration with weak relay-destination channel. }
\label{Unsymmetric_Bad_relay}
\end{figure}

In Fig. \ref{Aggregate_Assymatric_Bad_relay}, we use the same
system parameters as in the former figure. We plot the maximum
overall stable throughput versus $w_1$, which is the fraction of
time allocated for $s_1$. This figure shows that as $w_1$
increases the maximum aggregate stable throughput of the network
decreases. The reason behind this is that the maximum stable
throughput for $s_1$ is equal to $0.035$ while that for $s_2$ is
$0.9$. Consequently, as we allocate more time slots for $s_2$, the
aggregate stable throughput of the network increases. Moreover, as
shown in Fig. \ref{Unsymmetric_Bad_relay}, the two proposed
schemes, SBC and DBC, can achieve the same maximum stable
throughput for $s_2$ which dominates the aggregate throughput for
any $w_1$, and this is the reason why the aggregate throughput for
the proposed schemes is close in Fig.
\ref{Aggregate_Assymatric_Bad_relay}.

\begin{figure}[t]
  \centering
\includegraphics[width=.95\linewidth,height=.25\textheight]{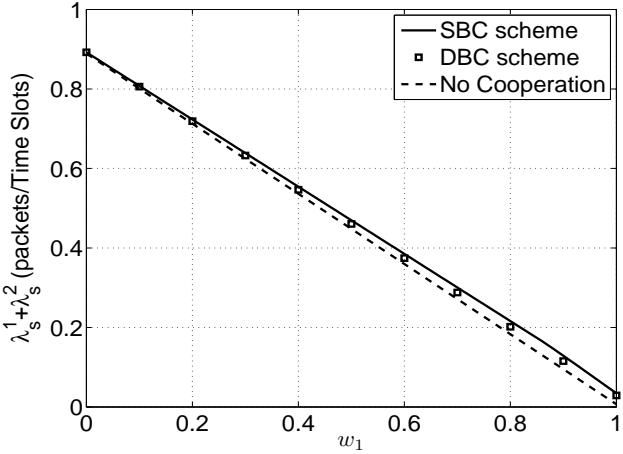}
\caption{ The maximum aggregate stable throughput for two source terminals in asymmetric channels configuration and weak relay-destination channel. }
\label{Aggregate_Assymatric_Bad_relay}
\end{figure}

In the second case, Fig. \ref{Unsymmetric_good_relay}, we plot the
stability region for different cooperative schemes for the
following system parameters: $ P{=}10, R{=}1$,
$\rho_{s_1,d}^{2}{=}0.8$, $\rho_{s_2,d}^{2}{=}0.08$,
$\rho_{s_1,r}^{2}{=}0.85$, $\rho_{s_2,r}^{2}{=}0.9$, and
$\rho_{r,d}^{2}{=}0.97$. This case corresponds to an asymmetric
channel situation, where the channels $s_1$-$d$ and $r$-$d$
are strong while the channel $s_2$-$d$ is weak. It is important to
note that the $r$-$d$ channel is stronger than that in the
former case. In this case both proposed cooperative schemes
achieve the same stability region, and hence, sensing does not
increase the stability region of the system. In the DBC scheme,
the relay transmits according a random experiment. If the relay
decides, randomly, to interfere with $s_1$, the destination can
decode both signals because both channels, $s_1$-$d$ and $r$-$d$,
are strong. Alternatively, if the relay interferes with $s_2$, the
destination most probably decodes the relay signal first, since
$r$-$d$ channel is strong, then decodes the source signal.
Consequently, the MPR capability at the destination decreases the
need to detect idle time slots, and hence, both proposed schemes
can achieve the same stability region. Removing the MPR capability
at the destination, in absence of sensing, yields to the
performance of the CM-DBC scheme which is much worse than that of
the DBC scheme.

It is important to notice that both proposed schemes outperform
the CCMA scheme which restricts the relay to send only in the idle
time slots. Another insight from this figure is the importance of
$\beta_{12}$ and $\beta_{21}$ in the SBC scheme and $\alpha_{12}$
and $\alpha_{21}$ in the DBC scheme to achieve this stability
region. These probabilities have crucial effect in the asymmetric
channel case as they enable the relay to utilize the time slots of
one source to transmit the traffic of the other. Note that most of
the packets of $s_1$ are transmitted directly to the destination
due to the high gain of its direct channel. In contrast, $s_2$
suffers from low direct channel gain, so most of its packets are
relayed. In the proposed schemes, the relay has the capability to
interfere with $s_1$, which has high direct channel gain, and send
a relayed packet of $s_2$, who suffers from low direct channel
gain. This capability expands the stability region for both
schemes.

\begin{figure}[t]
  \centering
\includegraphics[width=.95\linewidth,height=.25\textheight]{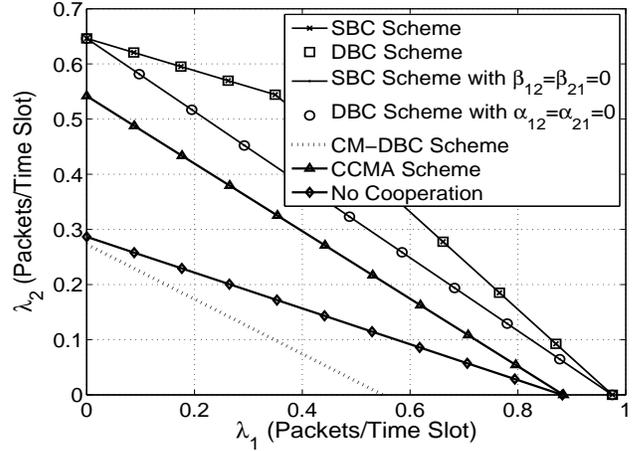}
\caption{ Stable throughput region for asymmetric channels configuration with strong relay-destination channel. }
\label{Unsymmetric_good_relay}
\end{figure}

In Fig. \ref{Aggregate_Assymatric_good_relay}, we plot the
aggregate stable throughput of the two source network versus
$w_1$. We use the same system parameters as in Fig.
\ref{Unsymmetric_good_relay}. It is obvious that the proposed
schemes achieve higher aggregate throughput than that of the
CCMA scheme. Unlike Fig. \ref{Aggregate_Assymatric_Bad_relay}, the
aggregate stable throughput increases as we allocate more time slots for
$s_1$. In this case, the maximum stable throughput of $s_1$,
$\lambda_s^1{=}0.98$, is greater than that of $s_2$,
$\lambda_s^2{=}0.65$. Consequently, as $w_1$ increases the
aggregate throughput of the network increases.

\begin{figure}[t]
  \centering
\includegraphics[width=.95\linewidth,height=.25\textheight]{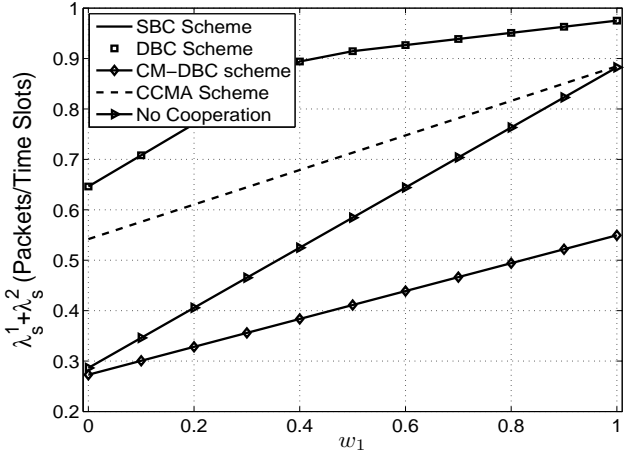}
\caption{ The maximum aggregate stable throughput for two source terminals in asymmetric channels configuration and strong relay-destination channel. }
\label{Aggregate_Assymatric_good_relay}
\end{figure}

In the third case, Fig. \ref{Nealy_symmatric}, we plot the
stability region\footnote{We do not add the DBC scheme in the
figure, because it is easy to realize from Fig.
\ref{Unsymmetric_good_relay} that the performance of the DBC scheme in
this case is exactly the same as that of the SBC scheme.} for a roughly
symmetric channel configuration where both direct links, $s_1$-$d$
and $s_2$-$d$, are strong. The system parameters are chosen as
follows: $\rho_{s_1,d}^{2}{=}0.75$, $\rho_{s_2,d}^{2}{=}0.8$,
$\rho_{s_1,r}^{2}{=}0.63$, $\rho_{s_2,r}^{2}{=}0.73$, and
$\rho_{r,d}^{2}{=}0.85$. The SBC scheme still exceeds the CCMA
scheme, however, in this case the role of of $\beta_{12}$ and
$\beta_{21}$ diminishes due to the symmetric configuration. In
this case, both source terminals have the same channel condition
so the relay does not need to assist one source terminal at the
expense of the other. Another insight from this figure is that,
when the direct link channels are strong, we can achieve the
performance of the CCMA scheme without even cooperation. However,
the proposed schemes expand the stability region % of the CCMA scheme 
even when both direct links are strong.

\begin{figure}[t]
  \centering
\includegraphics[width=.95\linewidth,height=.25\textheight]{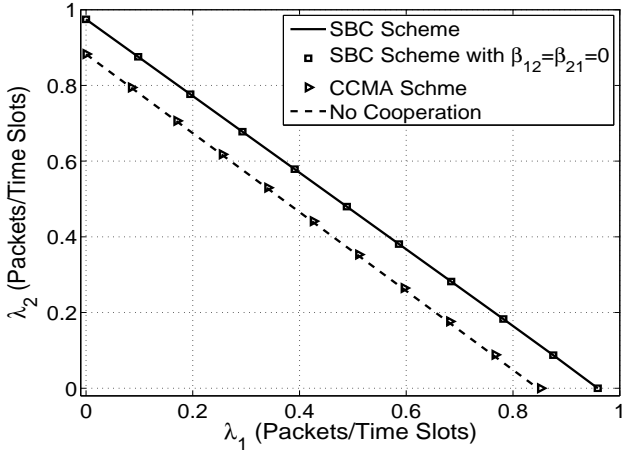}
\caption{ Stable throughput region for nearly symmetric channels configuration. }
\label{Nealy_symmatric}
\end{figure}

In Fig. \ref{spectral}, we illustrate how
the proposed schemes respond when the channel capability to
tolerate interference changes. We depict
the effect of the transmission rate $R$ on the maximum  aggregate
stable throughput of the network for symmetric channel
configuration. The parameters are chosen as follows: $w{=}0.5$,
$\rho_{s_1,d}^{2}{=}0.8$, $\rho_{s_2,d}^{2}{=}0.8$,
$\rho_{s_1,r}^{2}{=}0.95$, $\rho_{s_2,r}^{2}{=}0.95$, and
$\rho_{r,d}^{2}{=}0.96$. It is clear that the stable throughput of
the SBC and DBC schemes decrease slower than the CCMA and CM-DBC
schemes. At low transmission rates, the channels can tolerate the
interference, consequently, detecting the idle time slots is not
important. It is obvious from the figure that the SBC and DBC
schemes have the same performance for low transmission rates. As
the transmission rate increases, the capability to sustain
interference for all wireless channels decreases, and hence, it
becomes essential for the relay to transmit only in the idle time
slots. In the figure, as the transmission rate increases the
performance of the two proposed schemes approaches to that of the
CCMA scheme because the values of $\{ \beta_{ij} \}_{i,j{=}1}^2 $
and $\{ \alpha_{ij} \}_{i,j{=}1}^2$ begin to decrease to limit the
negative effect of interference. As we increase the transmission
rate more, the channels can not sustain any interference. The
values of $\{ \beta_{ij} \}_{i,j{=}1}^2 $ in the SBC scheme
diminish to be almost zeros, and the relay is restricted to send
only in the idle time slots. This means that the SBC scheme boils
down to the CCMA scheme. On the other hand, for the DBC scheme, as
the transmission rate increases, the ability that the destination
decodes the transmission of two nodes simultaneously decreases,
and  hence, the DBC scheme turns to the CM-DBC scheme.

\begin{figure}[t]
  \centering
\includegraphics[width=.95\linewidth,height=.25\textheight]{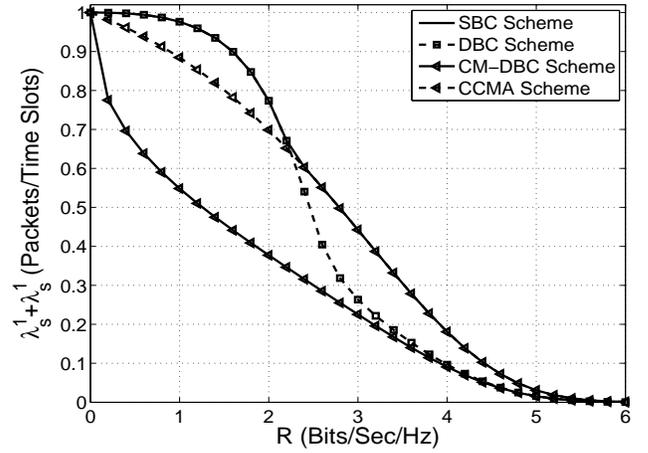}
\caption{ The maximum aggregate stable throughput versus the
transmission rate for symmetric
channels configuration. } \label{spectral}
\end{figure}

In Figs. \ref{Delay1} and \ref{Delay_unsymmatric}, we illustrate
the delay performance of the proposed schemes. First, in Fig.
\ref{Delay1}, we plot the minimum average delay encountered by the
packets of $s_2$, under the constraint that
$\lambda_s^1{=}0.29$. Moreover, we use the same system parameters
as in Fig. \ref{Unsymmetric_Bad_relay}, where the two proposed schemes do not achieve the same stability region. The maximum stable
throughputs for $s_2$, when $\lambda_s^1{=}0.29$, in the SBC and
DBC schemes are $0.83$ and $0.77$, respectively. We
do not plot the CCMA scheme to have a clear comparison
between the plotted schemes, since the CCMA scheme
performance is way worse than both. It is clear from
the figure that the delay performance of the two proposed schemes
in this case is close to each other for low $\lambda_2$. However,
the SBC scheme delay performance slightly exceeds that of the DBC
scheme. As $\lambda_2$ increases the SBC scheme begins to
significantly outperform the DBC scheme. It is obvious that the
results obtained through queue simulation are very close to the
results of the closed-form expressions derived in
(\ref{average_delay_enhance}) and (\ref{average_delay_DBC}). This
validates the soundness of the mathematical model. Moreover, the
trade-off between the stable throughput and the average delay is
clear where as the throughput increases the delay increases.

\begin{figure}[t]
  \centering
\includegraphics[width=.95\linewidth,height=.25\textheight]{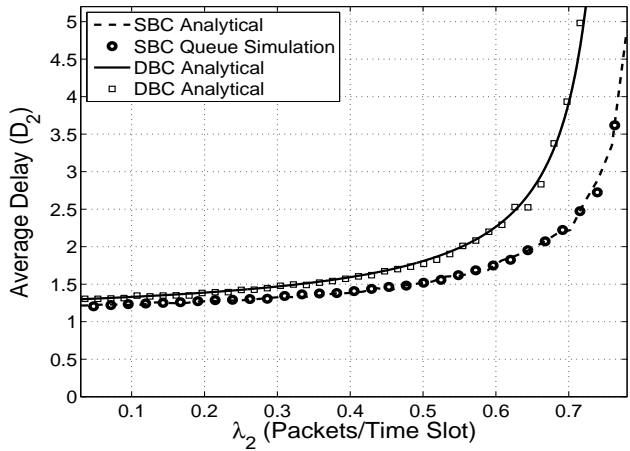}
\caption{ The average delay encountered by the packets of $s_2$ for asymmetric channels configuration with weak relay-destination channel.  } \label{Delay1}
\end{figure}

In Fig. \ref{Delay_unsymmatric}, we plot the minimum delay for
$s_1$ when $\lambda_s^2{=}0.81$. We use the same system
parameters as those in Fig. \ref{Unsymmetric_good_relay}, where the two proposed schemes achieve the same stability region. In this case, the maximum stable throughput for $s_2$ in the SBC, DBC and
CCMA schemes are $0.13$, $0.13$, and $0.038$, respectively. The
figure depicts that the proposed schemes significantly outperform the
CCMA scheme. The rationale behind this is that the CCMA scheme
allocates a large fraction of the time slots for $s_2$ to satisfy
the constraint $\lambda_s^2{=}0.81$, besides, the relay transmits
only in the idle time slots. Consequently, the delay performance
of the CCMA is much worse than that of the two proposed schemes.
It is worth noting that even when the two proposed schemes can
achieve the same stability region, the delay performance of the
SBC scheme outperforms that of the DBC scheme. This is because the
SBC scheme exploits the idle time slots to transmit the relayed
packets, and the capability of the channel to sustain the
interference to send simultaneously with the source terminal.
Consequently, the SBC scheme exceeds the DBC scheme that exploits
only the capability of the channels to tolerate interference.

\begin{figure}[t]
  \centering
\includegraphics[width=.95\linewidth,height=.25\textheight]{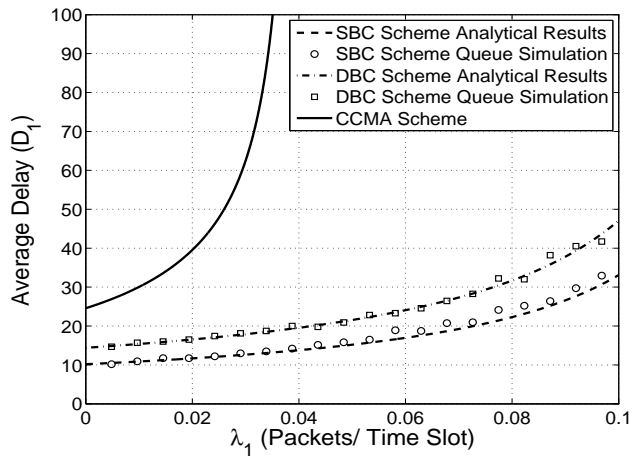}
\caption{ The average delay encountered by the packets of $s_1$ for asymmetric channels configuration with strong relay-destination channel. } \label{Delay_unsymmatric}
\end{figure}

\section{Sensing at the relay versus MPR at the destination}
\label{discussion}

In this section, we illustrate how the MPR capability at the
destination can compensate for the lack of sensing at the relay.
Moreover, we derive the condition under which the two proposed
schemes achieve exactly the same maximum stable throughput,
$\lambda_s^i$. In this part, we assume that $w_1{=}1$, and hence,
all time slots are allocated to $s_1$, i.e., we have only one
source terminal $s_1$.

\newtheorem{theorem}{Theorem}
\begin{theorem}
\label{th1}
The maximum stable throughput for the two proposed schemes is given by
\begin{equation}
\label{max_opt}
\lambda_{s}^{1^*}=(1{-}\frac{T_1}{T_1{+}g_{rd}^{s_1}})(f_{s_1d}{+}T_1){ +} \frac{T_1}{T_1{+}g_{rd}^{s_1}} g_{s_1d}^r
\end{equation}
if the following condition is satisfied
\begin{align}
\label{cond_opt}
f_{rd}-g_{rd}^{s_1}& \leq \min \bigg\{ \frac{g_{s_1d}^r(T_1+f_{rd})}{T_1+f_{s_1d}}, \\ \nonumber &\frac{(T_1+g_{rd}^{s_1})^2(f_{s_1d}+T_1)}{T_1(f_{s_1d}+T_1-g_{s_1d}^r)}{-}(T_1+2g_{rd}^{s_1})\bigg\}
\end{align}
where $T_1{=}(1{-}f_{s_1d})f_{s_1r}$.
\end{theorem}

\begin{proof}[Proof of Theorem 1]
In the SBC scheme, we can rewrite the stability condition in (\ref{max_throughput_SBC}) for $w_1{=}1$ as
\begin{equation}
\lambda_{s}^1< \min\{ \mu_1^{\text{SBC}},\mu_{u_1}^{\text{SBC}} \}
\end{equation}
where\begin{align}
\mu_1^{\text{SBC}}&= (1-\beta_{11})(f_{s_1d}+T_1) + \beta_{11} g_{s_1d}^r \\
\mu_{u_1}^{\text{SBC}}&= \frac{f_{rd}}{(1-\beta_{11})T_1+f_{rd}-\beta_{11}g_{rd}^{s_1} } \mu_1^{\text{SBC}}
\end{align}
The optimization problem in (\ref{opt1}) can be written as
\begin{equation}
\label{opt_simple_SBC}
\begin{aligned}
& \underset{\beta_{11} }{\text{max}}
&&  \min\{\mu_1^{\text{SBC}},\mu_{u_1}^{\text{SBC}}\}\\
& \text{subject to}
&&  0 \leq  \beta_{11} \leq 1
\end{aligned}
\end{equation}
\newtheorem{lema1}{Lemma}

\begin{lema1}
If the following condition is satisfied
\begin{equation}
\label{cond1_lema1} f_{rd}-g_{rd}^{s_1} \leq
\frac{g_{s_1d}^r(T_1+f_{rd})}{T_1+f_{s_1d}},
\end{equation}
%then the optimum value of $\beta_{11}$ for the 
the solution of the optimization problem in (\ref{opt_simple_SBC}) is given by
%the maximum stable throughput for $s_1$ in the SBC scheme %and the optimum value of $\beta_{11}$ for the optimization problem in (\ref{opt_simple_SBC}) are
%is given by
\begin{equation}
\lambda_{s}^{1^*}=(1{-}\frac{T_1}{T_1{+}g_{rd}^{s_1}})(f_{s_1d}{+}T_1){ +} \frac{T_1}{T_1{+}g_{rd}^{s_1}} g_{s_1d}^r 
\end{equation}

\end{lema1}

\begin{proof} See Appendix B
\end{proof}

For the DBC scheme, we  can rewrite the the optimization problem
in (\ref{opt2_DBC}) as follows
\begin{equation}
\label{opt_simple_DBC}
\begin{aligned}
& \underset{\alpha_{11} }{\text{max}}
&&  \min\{\mu_1^{\text{DBC}},\mu_{u_1}^{\text{DBC}}\}\\
& \text{subject to} &&  0 \leq  \alpha_{11} \leq 1
\end{aligned}
\end{equation}
where
\begin{align}
\label{mu_DBC}
\mu_1^{\text{DBC}}    &= (1-\alpha_{11})(f_{s_1d}+T_1) {+} \alpha_{11} g_{s_1d}^r \\
\label{mu_u1_DBC} \mu_{u_1}^{\text{DBC}}&=
\frac{\alpha_{11}f_{rd}}{(1-\alpha_{11})T_1+\alpha_{11}f_{rd}-\alpha_{11}g_{rd}^{s_1}
} \mu_1^{\text{DBC}}
\end{align}

\begin{lema1}
If the following condition is satisfied
\begin{equation}
\label{cond2_lema2}f_{rd}-g_{rd}^{s_1}\leq \frac{(T_1+g_{rd}^{s_1})^2(f_{s_1d}+T_1)}{T_1(f_{s_1d}+T_1-g_{s_1d}^r)}{-}(T_1+2g_{rd}^{s_1})
\end{equation}
%the optimum value of $\alpha_{11}$ for the optimization problem (\ref{opt_simple_DBC}) is given by
the maximum stable throughput for $s_1$ in the DBC scheme, which is the solution of the problem in (\ref{opt_simple_DBC}), is given by
\begin{equation}
\lambda_{s}^{1^*}=(1{-}\frac{T_1}{T_1{+}g_{rd}^{s_1}})(f_{s_1d}{+}T_1){ +} \frac{T_1}{T_1{+}g_{rd}^{s_1}} g_{s_1d}^r
\end{equation}
\begin{proof}[Proof]
See Appendix C
\end{proof}
\end{lema1} The conditions (\ref{cond1_lema1}) and (\ref{cond2_lema2}) establish the result in (\ref{max_opt}) and (\ref{cond_opt}).

%It worth noting that, if we apply the same analysis for any source terminal $s_i$, i.e. $w_i{=}1$, we will obtain the same condition as follows
%\begin{align}
%\label{cond_opt_general}
%f_{rd}-g_{rd}^{s_i}& \leq \min \bigg\{ \frac{g_{s_id}^r(T_i+f_{rd})}{T_i+f_{s_id}}, \\ \nonumber &\frac{(T_i+g_{rd}^{s_i})^2(f_{s_id}+T_i)}{T_i(f_{s_id}+T_i-g_{s_id}^r)}{-}(T_i+2g_{rd}^{s_i})\bigg\}
%\end{align}
%where $T_i{=}(1{-}f_{s_id})f_{s_ir}$.

\end{proof}

Theorem \ref{th1} states that if the values of $f_{rd}$ and $g_{rd}^{s_1}$ are close, i.e., they satisfy the condition in (\ref{cond_opt}), the two proposed schemes achieve exactly the same maximum throughput. The values $f_{rd}$ and $g_{rd}^{s_1}$ depend on the variance of two channels; $s_1$-$d$ and $r$-$d$. From the definitions in (\ref{f}) and (\ref{g}), the value of $g_{rd}^{s_1}$ is close to that of $f_{rd}$ in two cases. First, if the two channels, $s_1$-$d$ and $r$-$d$, are strong, i.e., the destination can decode both signals with high probability. Second, if at least one of the channels, $r$-$d$ or $s_1$-$d$, is strong, i.e., the destination can decode the strong signal first then the weak one.

%Although, we derive the condition where the two proposed schemes achieve the same maximum stable throughput in especial case, for $w_i${=}1, , the insights from this condition match that in Figs \ref{Unsymmetric_Bad_relay} and \ref{Unsymmetric_good_relay}. In Fig. \ref{Unsymmetric_Bad_relay}, the maximum stable throughput of $s_1$ in the SBC scheme equals to $0.035$ which is greater than that in the DBC scheme $0.029$, because both channels, $r$-$d$ and $s_1$-$d$, are weak. Alternatively, in the same figure, the maximum stable throughput of $s_2$ in the SBC scheme is exactly the same as that in the DBC scheme because we have one strong channel $s_2$-$d$. Moreover, the two proposed schemes achieve the same stability region. % This strong channel facilitates the decoding mission at the destination.

We can easily map the insights from the obtained condition in (\ref{cond_opt}) to that in Fig. \ref{Unsymmetric_Bad_relay} and Fig. \ref{Unsymmetric_good_relay}. First, in Fig. \ref{Unsymmetric_Bad_relay}, the condition in (\ref{cond_opt}) is violated for $s_1$, where we have two weak channels, $s_1$-$d$ and $r$-$d$. Hence, the maximum stable throughput of $s_1$ in the SBC scheme %equals to $0.035$ which
is greater than that in the DBC scheme. Alternatively, in the same figure, the condition is satisfied for $s_2$, where we have one strong channel $s_2$-$d$. Thus, the maximum stable throughput of $s_2$ in the SBC scheme is exactly the same as that in the DBC scheme. Second, in Fig. \ref{Unsymmetric_good_relay}, the condition is satisfied for both users, hence, the maximum stable throughput for $s_1$ and $s_2$ is the same for the two proposed schemes. Moreover, the two schemes achieve the same stability region. From these results, we can realize that the MPR capability at the destination can compensate for the need for sensing to detect the idle time slots, if there is at least one strong channel to the destination. The strong channel facilitates the decoding at the destination, and mitigates the need of the relay to send only in empty channels.

\begin{figure}[t]
  \centering
\includegraphics[width=.95\linewidth,height=.25\textheight]{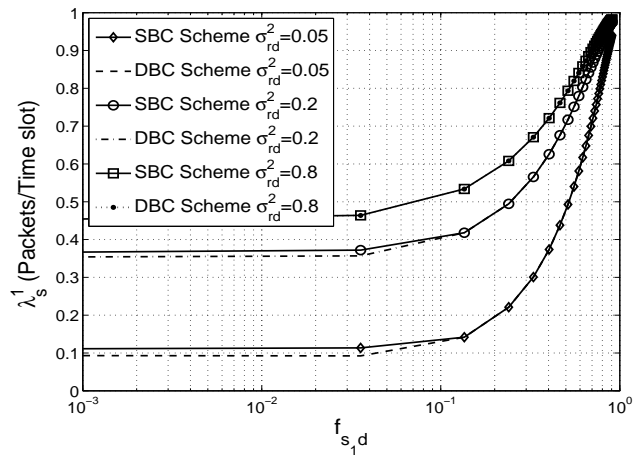}
\caption{ The effect of the channels, $s_1$-$d$ and $r$-$d$, on the maximum stable throughput of the two proposed schemes for $w_1{=}1$ for three fixed values for $r$-$d$ channel variance. }
\label{Singlr_User}
\end{figure}

In Fig. \ref{Singlr_User}, we show numerically different channel variances that satisfy the condition in (\ref{cond_opt}). We plot the maximum stable throughput of $s_1$ versus $f_{s_1d}$ by varying $\sigma_{s_1d}^2$ from zero to one for three fixed values $\sigma_{rd}^2$. The system parameters of this figure are chosen as follows: $w_1{=}1$, $ P{=}10$, $R{=}1$, and $\rho_{s_1,r}^{2}{=}0.8$. In the first case, where $\sigma_{rd}^2{=}0.05$, the two proposed schemes achieve the same stable throughput approximately at $f_{s_1d}{=}0.15$. As we increase $\sigma^2_{rd}$, the performance of the two proposed schemes becomes closer than that of the former case. Ultimately, for the case of strong $r$-$d$ channel $\sigma_{rd}^2{=}0.8$, the two proposed schemes almost achieve the same performance for any $\sigma_{s_1d}^2$, and this result emphasises the insights obtained from the condition in (\ref{cond_opt}).

\section{Conclusion}
\label{conclusion}

In this paper, we have proposed two cooperation schemes, and studied their impact at the medium access layer metrics such as stable throughput and average delay. In the SBC scheme, the relay senses the channel at the beginning of each time slot, and it decides either to transmit or receive packets depending on the sensing outcome. On the other hand, in the DBC scheme, the relay does not sense the channel, and it decides its operation in a random fashion. For each scheme, we derived the stability conditions for each queue in the system and characterized the stability region. Furthermore, we derived
approximate expression for the average delay encountered by the
packets. We illustrated how the SBC scheme significantly
outperforms existing cooperative schemes. The SBC scheme exploits the
available resources more efficiently than other cooperative
schemes because the relay not only utilizes the idle time slots, but also interferes with the source terminals in a mild way to
mitigate the adverse effects of interference. %We show that the DBC scheme can achieve almost the same stability region as that of the SBC scheme under certain condition, and this means that the role of sensing at the relay is diminished.

Moreover, we demonstrated that the MPR capability at the
destination can compensate for the relay need to sense the channel. 
Although the relay in the DBC scheme does not sense the channel,
our results show that the DBC scheme can achieve,
under a certain condition, the same stability region as that of the
SBC scheme. %and this means that the role of sensing at the relay is diminished. We, also, derived the condition under which the two

\section*{Appendix A}
\section*{Derivation of $g_{mn}^I$}

The term $g_{mn}^I$ denotes the probability that the link ($m,n$) is not in
outage in presence of interference from node I. Thus, we can express it as follows
\begin{equation}
\label{g}
g_{mn}^I=v_{mn}^I+(1-v_{mn}^I)h_{mn}^I
\end{equation}
where $v_{mn}^I$ is the probability that the node $n \in L$ successfully decodes both packets transmitted from $m \in T$ and $I \in T$, on the other hand, $h_{mn}^I $ is the probability that the node $n$ successfully decodes the packet transmitted from $m$ by treating $I$ as noise. Let X and Y be two independent exponential random variables with parameters $\gamma_1$ and $\gamma_2$, respectively, and $p(x)$ and $p(y)$ be their probability density functions. We define two deterministic variables $\eta$ and $\eta_1$. To derive the expression of $v_{mn}^I$, we first define the region
\begin{align}
\Re(\eta,\eta_1)&= \{(x,y): x>\eta \cap y>  \cap x{+}y>\eta_1\} \\ \nonumber
                &=\{(x,y) : x>\eta \cap y > \max[\eta,\eta_1-x]\} \\ \nonumber
                &= \{(x,y): \eta_1-\eta \geq x \geq \eta \cap y > \eta_1 -x \} \\ \nonumber
                &\cup \{(x,y): x> \eta_1-\eta \cap y>\eta \},
\end{align}
then, we figure the following integration
\begin{align}
V(\eta,\eta_1,\gamma_1,\gamma_2){=}& \iint\limits_{\Re(\eta,\eta_1)} p(x,y) dx dy \\ \nonumber
                               {=}& \int\limits_\eta^{\eta_1-\eta} p(x) \int\limits_{\eta-x}^\infty p(y) \enspace dy dx\\ \nonumber
 {+}& \int\limits_{\eta_1-\eta}^\infty p(x)\int\limits_\eta^\infty p(y) \enspace dy dx \\ \nonumber
 {=}&   \frac{\gamma_1e^{\gamma_2 \eta_1}}{\gamma_1{-}\gamma_2}  \resizebox{.58\hsize}{!}{$ (\exp\big({-}(\gamma_1{-}\gamma_2)\eta\big){-}\exp({-}(\gamma_1{-}\gamma_2)(\eta_1{-}\eta)))$} \\ \nonumber
 {+}& \exp(\eta_1 \gamma_1 +\eta(\gamma_2{-}\gamma_1))
\end{align}
Hence, we have
\begin{equation}
v_{mn}^I= V(\frac{2^R{-}1}{P},\frac{2^{2R}{-}1}{P}, \frac{1}{\rho_{m,n}^{2}},\frac{1}{\rho_{I,n}^{2}} ).
\end{equation}

To derive the expression of $h_{mn}^I$, we first perform the following integration
\begin{align}
H(\eta,\gamma_1,\gamma_2)&= \mathbb{P} \{ R<\log(1+\frac{Px}{Py+1}) \} \\ \nonumber                         &= \int\limits_0^\infty \exp\big(\frac{(2^R-1)(Py+1)\gamma_1}{P}\big) p(y) dy \\ \nonumber
                         &= \frac{\gamma_2 \exp(-\eta \gamma_1)}{\gamma_2+P \gamma_1 \eta}.
\end{align}
Hence, the probability $h_{mn}^I$ is given by
\begin{equation}
h_{mn}^{I}= H(\frac{2^R-1}{P},\frac{1}{\rho_{m,n}^2},\frac{1}{\rho_{I,n}^2})
\end{equation}

\section*{Appendix B}
\section*{Proof of Lemma 1}
%\begin{proof}[Proof of Lemma 1]
Taking the derivative of $\mu_1^{\text{SBC}}$ and $\mu_{u_1}^{\text{SBC}}$ with respect to
$\beta_{11}$ yields the following
\begin{align}
\label{dmu1}
\frac{\partial \mu_1^{\text{SBC}}}{\partial \beta_{11}} &= g_{s_1d}^r-f_{s_1d}-T_1 \\
\label{dmuu1} \frac{\partial \mu_{u_1}^{\text{SBC}}}{\partial \beta_{11}}&{=}
\frac{f_{rd}
(g_{s_1d}^r{-}T_1{-}f_{s_1d})(T_1{+}f_{rd}{-}\beta_{11}(T_1{+}g_{rd}^{s_1}))}{\big((1-\beta_{11})T_1+f_{rd}-\beta_{11}g_{rd}^{s_1}\big)^2}\\
\nonumber &
+\frac{f_{rd}(T_1+g_{rd}^{s_1})(T_1{+}f_{s_1d}{+}\beta_{11}(g_{s_1d}^r+T_1-f_{s_1d}))}{\big((1-\beta_{11})T_1+f_{rd}-\beta_{11}g_{rd}^{s_1}\big)^2}
\end{align}
From (\ref{dmu1}), we can see that $\mu_1^{\text{SBC}}$ is a monotonically
decreasing function in $\beta_{11}$ because, from definition,
$f_{s_1d}$ is greater than $g_{s_1d}^{r}$. On the other side, from
(\ref{dmuu1}), we can show that $\mu_{u_1}^{\text{SBC}}$ is monotonically
increasing in $\beta_{11}$ if the following condition is satisfied
\begin{equation}
\label{cond1} f_{rd}-g_{rd}^{s_1} \leq
\frac{g_{s_1d}^r(T_1+f_{rd})}{T_1+f_{s_1d}}.
\end{equation}
Since the value of $\mu_1^{\text{SBC}}$ at $\beta_{11}{=}{0}$ is greater than
the value of $\mu_{u_1}^{\text{SBC}}$ at $\beta_{11}{=}{0}$ and $\mu_1^{\text{SBC}}$
decreases monotonically with $\beta_{11}$ while $\mu_{u_1}^{\text{SBC}}$
increases when (\ref{cond1}) is satisfied, the optimal solution of
(\ref{opt_simple_SBC}) occurs when $\mu_{1}^{\text{SBC}}=\mu_{u_1}^{\text{SBC}}$. Thus,
under the above condition in (\ref{cond1}), the optimum
$\beta_{11}$ is given by
\begin{equation}
\beta_{11}^*= \frac{T_1}{T_1+g_{rd}^{s_1}}
\end{equation}
and, hence, the maximum stable throughput for $s_1$ is given by
\begin{equation}
\label{max_throu_SBC}
\lambda_{s}^{1^*}=(1{-}\frac{T_1}{T_1{+}g_{rd}^{s_1}})(f_{s_1d}{+}T_1){
+} \frac{T_1}{T_1{+}g_{rd}^{s_1}} g_{s_1d}^r
\end{equation}
%\end{proof}

\section*{Appendix C}
\section*{Proof of Lemma 2}
It is clear, from (\ref{mu_DBC}), that $\mu_1^{\text{DBC}}$ decreases monotonically with $\alpha_{11}$. On the other hand, % $\mu_{u_1}$ is quadratic over linear function in $\alpha_{11}$, and hence, it is a concave function in $\alpha_{11}$ for the range of $\alpha_{11}$ from zero to one \cite{boyd2004convex}.
the derivative of $\mu_{u_1}^{\text{DBC}}$ with respect to $\alpha_{11}$ is given by
\begin{align}
\frac{\partial \mu_{u_1}^{\text{DBC}}}{\partial \alpha_{11}}&{=}\resizebox{.9\hsize}{!}{$  \frac{ f_{rd}(T_1{+}f_{s_1d}{+}2\alpha_{11}(g_{s_1d}^r{-}T_1{-}f_{s_1d}))(T_1{+} \alpha_{11}(f_{rd}{-}T_1{-}g_{rd}^{s_1}))}  {((1-\alpha_{11})T_1+\alpha_{11}f_{rd}-\beta_{11}g_{rd}^{s_1} )^2} $} \nonumber \\
& {-} \resizebox{.9\hsize}{!}{$ \frac{f_{rd}( f_{rd}-T_1-g_{rd}^{s_1} )(\alpha_{11}(f_{s_1d}+T_1) +\alpha_{11}^2(g_{s_1d}^r-f_{s_1d}-T_1))}{((1-\alpha_{11})T_1+\alpha_{11}f_{rd}-\alpha_{11}g_{rd}^{s_1} )^2} $}
\end{align}
It is obvious from (\ref{mu_u1_DBC}) that $\mu_{u_1}^{\text{DBC}}$ is a quadratic over linear function in $\alpha_{11}$, and hence, it is a concave function in $\alpha_{11}$ for the range of $\alpha_{11}$ from zero to one \cite{boyd2004convex}.
Besides, $\mu_{u_1}^{\text{DBC}}$ increases at $ \alpha_{11}{=}0$, because the slope of $\mu_{u_1}^{\text{DBC}}$ at $\alpha_{11}{=}0$ is positive and given by
\begin{equation}
\frac{f_{rd}T_1(T_1+f_{s_1d})}{((1-\alpha_{11})T_1+\alpha_{11}f_{rd}-\alpha_{11}g_{rd}^{s_1} )^2}
\end{equation}
%The optimum value of $\alpha_{11}$ is obtained at the intersection point between $\mu_1$ and $\mu_{u_1}$ if the slope of $\mu_{u_1}$ is positive at this intersection point because of the following
Note that the two functions $\mu_1^{\text{DBC}}$ and $\mu_{u_1}^{\text{DBC}}$ satisfy the following three conditions. First, the value of $\mu_1^{\text{DBC}}$ at $\alpha_{11}={0}$ is greater than the value of $\mu_{u_1}^{\text{DBC}}$ at the same point. Second, $\mu_1^{\text{DBC}}$ is monotonically decreasing function in $\alpha_{11}$. Third, $\mu_{u_1}^{\text{DBC}}$ is concave function with positive slope at $\alpha_{11}{=}0$. Therefore, the optimum value of $\alpha_{11}$ is obtained at the intersection point between $\mu_1^{\text{DBC}}$ and $\mu_{u_1}^{\text{DBC}}$ if the slope of $\mu_{u_1}^{\text{DBC}}$ is positive at this point because this means that $\mu_{u_1}^{\text{DBC}}$ increases monotonically with $\alpha_{11}$ from $\alpha_{11}{=0}$ until the intersection point. $\mu_{u_1}^{\text{DBC}}$ has a positive slope at the intersection point if the following condition satisfied
\begin{equation}
\label{cond2}
f_{rd}-g_{rd}^{s_1}\leq \frac{(T_1+g_{rd}^{s_1})^2(f_{s_1d}+T_1)}{T_1(f_{s_1d}+T_1-g_{s_1d}^r)}{-}(T_1+2g_{rd}^{s_1})
\end{equation}
Under the above condition, the optimum solution of (\ref{opt_simple_DBC}) occurs when $\mu_1^{\text{DBC}}{=}\mu_{u_1}^{\text{DBC}}$, hence the optimum value of $\alpha_{11}$ is given by
\begin{equation}
\alpha_{11}^*= \frac{T_1}{T_1+g_{rd}^{s_1}}
\end{equation}
and, the maximum stable throughput for $s_1$ is given by
\begin{equation}
\lambda_{s}^{1^*}=(1{-}\frac{T_1}{T_1{+}g_{rd}^{s_1}})(f_{s_1d}{+}T_1){ +} \frac{T_1}{T_1{+}g_{rd}^{s_1}} g_{s_1d}^r
\end{equation}
which is exactly the same stable throughput for the SBC scheme in (\ref{max_throu_SBC}).

\bibliographystyle
{IEEEtran}
\bibliography{IEEEabrv,edit_draft}
\end{document}